\documentclass[11pt]{article}

\usepackage[T1]{fontenc}
\usepackage[utf8]{inputenc}
\usepackage{lmodern}        
\usepackage{iftex}          
\ifPDFTeX                   
  \usepackage[activate={true,nocompatibility},final]{microtype}
\else                       
  \usepackage[protrusion=true,final]{microtype}
\fi
\usepackage{times}

\usepackage[margin=1in]{geometry}

\usepackage{authblk}                    

\usepackage{amsmath}
\usepackage{amssymb}        
\usepackage{mathtools}      
\usepackage{nicefrac}
\allowdisplaybreaks         

\usepackage{amsthm}
\usepackage{bm}             
\usepackage{thm-restate}    

\usepackage{graphicx}
\graphicspath{{figures/}{./}}                       
\usepackage{booktabs}                               
\usepackage{multirow,makecell,array}
\usepackage[font=small,labelfont=bf]{caption}       
\usepackage{subcaption}                             
\usepackage{algorithm}
\usepackage{algpseudocode}                          

\usepackage{enumitem}
\setlist[itemize]{leftmargin=2.2em,itemsep=2pt,topsep=2pt}
\setlist[enumerate]{leftmargin=2.2em,itemsep=2pt,topsep=2pt}

\usepackage{xcolor}
\definecolor{LinkColor}{rgb}{0.10,0.40,0.75}        
\definecolor{CiteColor}{rgb}{0.70,0.25,0.20}        
\definecolor{UrlColor} {rgb}{0.20,0.50,0.50}        
\definecolor{TodoColor}{rgb}{0.80,0.30,0.10}        

\usepackage{tikz}
\usetikzlibrary{positioning,calc,arrows.meta}

\usepackage[round,sort&compress]{natbib}

\usepackage{url}
\usepackage{hyperref}
\hypersetup{
  colorlinks=true,
  linkcolor=LinkColor,
  citecolor=CiteColor,
  urlcolor=UrlColor,
  breaklinks=true,
  bookmarksnumbered=true,
}
\usepackage{bookmark}                               
\usepackage[capitalise,nameinlink,noabbrev,sort&compress]{cleveref}  

\numberwithin{equation}{section}

\newif\ifdraft \draftfalse
\ifdraft
  \usepackage{lineno}\linenumbers
  \newcommand{\todo}[1]{\textcolor{TodoColor}{\textbf{[TODO:}~#1\textbf{]}}}
\else
  \newcommand{\todo}[1]{}                            
\fi

\newcommand{\R}{\mathbb{R}}

\newcommand{\E}{\mathbb{E}}                       
\newcommand{\mc}[1]{\mathcal{#1}}                 

\newcommand{\eps}{\varepsilon}

\newcommand{\defeq}{\coloneqq}                    

\DeclareMathOperator{\tr}{tr}

\DeclarePairedDelimiter{\abs}{\lvert}{\rvert}
\DeclarePairedDelimiter{\norm}{\lVert}{\rVert}
\DeclarePairedDelimiter{\set}{\{}{\}}

\DeclarePairedDelimiterX{\inner}[2]{\langle}{\rangle}{#1,#2}   

\newcommand{\Poly}{P}                              
\newcommand{\Bone}{B_1^d}                           
\newcommand{\Binf}{B_\infty^d}                      
\newcommand{\Bp}{B_p^d}                             
\newcommand{\gauge}[2]{\lVert #1\rVert_{#2}}         
\newcommand{\ver}{v}                                
\newcommand{\Dt}{D_t}                               
\newcommand{\Mt}{M_t}                               
\newcommand{\simplex}{\Delta}                        
\newcommand{\dist}{\kappa}                            
\newcommand{\distnn}{\kappa^{\mathrm{nn}}}            
\newcommand{\designbd}{\beta}                         
\newcommand{\ratio}{\rho}                             
\newcommand{\Rad}{S}                                 
\newcommand{\kflat}[1]{x^{(#1)}}                     

\theoremstyle{plain}
\newtheorem{theorem}{Theorem}[section]

\newtheorem{fact}[theorem]{Fact}

\theoremstyle{definition}

\theoremstyle{remark}
\newtheorem{remark}[theorem]{Remark}

\crefname{theorem}{Theorem}{Theorems}          \Crefname{theorem}{Theorem}{Theorems}
\crefname{proposition}{Proposition}{Propositions}
\Crefname{proposition}{Proposition}{Propositions}
\crefname{lemma}{Lemma}{Lemmas}                \Crefname{lemma}{Lemma}{Lemmas}
\crefname{corollary}{Corollary}{Corollaries}   \Crefname{corollary}{Corollary}{Corollaries}
\crefname{conjecture}{Conjecture}{Conjectures} \Crefname{conjecture}{Conjecture}{Conjectures}
\crefname{fact}{Fact}{Facts}                   \Crefname{fact}{Fact}{Facts}
\crefname{definition}{Definition}{Definitions} \Crefname{definition}{Definition}{Definitions}
\crefname{assumption}{Assumption}{Assumptions} \Crefname{assumption}{Assumption}{Assumptions}
\crefname{example}{Example}{Examples}          \Crefname{example}{Example}{Examples}
\crefname{problem}{Problem}{Problems}          \Crefname{problem}{Problem}{Problems}
\crefname{remark}{Remark}{Remarks}             \Crefname{remark}{Remark}{Remarks}
\crefname{claim}{Claim}{Claims}                \Crefname{claim}{Claim}{Claims}
\crefname{algorithm}{Algorithm}{Algorithms}    \Crefname{algorithm}{Algorithm}{Algorithms}

\title{Low-Degree Polynomial Approximation of the Cross-Polytope}
\author{Xiaoyu Li\\
University of New South Wales\\
\texttt{xiaoyu.li2@unsw.edu.au}}
\date{}

\hypersetup{
  pdftitle={Low-Degree Polynomial Approximation of the Cross-Polytope},
  pdfauthor={Xiaoyu Li}
}

\begin{document}
\maketitle

\begin{abstract}
How much approximation power does polynomial degree buy for a symmetric convex body? Barvinok's lifted
ellipsoid construction shows that degree $2t$ always suffices for distortion $O(\sqrt{d/t})$ when $t\le d$,
but the construction itself does not identify a body that forces this curve throughout the degree range. We
show that the
cross-polytope does. If $\dist_t$ denotes the optimum over positive sum-of-squares forms and $\distnn_t$ the
optimum over all globally nonnegative forms, then
\[
  \frac1{\sqrt{2e}}\sqrt{\frac dt}
  \le \distnn_t(B_1^d)
  \le \dist_t(B_1^d)
  \le \sqrt{2e}\sqrt{\frac dt},
  \qquad 1\le t\le d.
\]
Thus degree $\Theta(d)$ is necessary and sufficient for constant distortion, even when the lower-bound
adversary may use any nonnegative form rather than an SoS certificate.

The lower proof is representation-free. Signed-permutation averaging removes orientation but still leaves a
large invariant algebra. Evaluating instead on the flat points of the $\ell_1$-sphere makes every invariant
monomial forget its partition structure and remember only the support size $k$. Every candidate then has the
profile $V(k)=k^{-2t}Q(k)$, where $\deg Q\le t$ and $Q(0)=0$; after the substitution $u=1/k$, Lagrange
interpolation shows that this degree budget cannot keep the profile nearly level across $d$ support scales. The
matching SoS construction is dictated by the same symmetry: averaging the powers of the sign-vector facet
normals gives the Rademacher form $\E_\varepsilon(\varepsilon^\top x)^{2t}$.

The mechanism also diagnoses what does and does not generalize. The polar cube has a constant certificate at
logarithmic degree, and a two-dimensional quartic shows that its familiar coordinate certificate need not be
optimal. For a general symmetric polytope, weighted facet powers yield a concave boundary-floor program whose
worst-direction oracle reduces to convex dual-norm problems; at degree two it recovers classical optimal design
and the John ellipsoid. This is an oracle-model certificate optimization, not an end-to-end complexity result.
Finally, a weighted reciprocal-grid lemma for $k^{-\alpha}Q(k)$ yields the $\ell_p$ consequences and identifies
$p=2$ as the precise point where the interpolation mechanism stops.

\end{abstract}

\section{Introduction}
\label{sec:intro}

John's theorem says that a quadratic form can approximate every origin-symmetric convex body in $\R^d$ with
loss at most $\sqrt d$. Allowing a homogeneous polynomial of degree $2t$ gives a sequence of increasingly rich
exact-degree algebraic models. The basic quantitative question is then: \emph{what does the additional degree
buy?} If degree is the available approximation budget, how quickly can the best worst-case distortion
available through degree $2t$ fall as $t$ grows?

There is one point of interpretation. The $2t$-th root of a nonnegative form is homogeneous, but it need not be
convex and hence need not be a norm. The quantity relevant here is only its multiplicative sandwich on the
boundary of the body. For a form $p$ that is positive away from the origin, write
\[
  \ratio_K(p^{1/(2t)})
  =
  \frac{\max_{\gauge{x}{K}=1}p(x)^{1/(2t)}}
       {\min_{\gauge{x}{K}=1}p(x)^{1/(2t)}}.
\]
We keep two polynomial cones separate:
\begin{equation}\label{eq:kappadef}
  \dist_t(K)=\inf_{\substack{p\in\Sigma_{d,2t}\\p(x)>0\text{ for }x\ne0}}\ratio_K(p^{1/(2t)}),
  \qquad
  \distnn_t(K)=\inf_{\substack{p\in\mc N_{d,2t}\\p(x)>0\text{ for }x\ne0}}\ratio_K(p^{1/(2t)}),
\end{equation}
where $\Sigma_{d,2t}$ is the sum-of-squares cone and $\mc N_{d,2t}$ is the cone of globally nonnegative forms.
Thus $\distnn_t(K)\le\dist_t(K)$. This distinction will matter in opposite directions: our upper
constructions are SoS, while the principal lower bound holds against the larger nonnegative cone.
Here $t$ specifies an exact homogeneous degree; the feasible classes at successive values of $t$ are not
literally nested, and we do not assume that $\dist_t(K)$ is monotone. When degree is meant as a cumulative
budget, one may instead take the best value over all $s\le t$; the preliminaries make this monotone envelope
explicit. The cross-polytope law below has the same order under either convention.

\paragraph{The universal curve and the missing worst case.}
Barvinok's polynomial approximation theorem, in its all-degree SoS form, gives
\begin{equation}\label{eq:barvinok}
  \dist_t(K)
  \le \binom{d+t-1}{t}^{1/(2t)}
  =O(\sqrt{d/t})
  \qquad (1\le t\le d)
\end{equation}
for every symmetric body~\citep{Barvinok2003ApproxNorm,BarvinokVeomett2008}. This tells us what every body can
achieve, but not whether any body forces us to pay that much degree.

The cube and the cross-polytope make the gap in our understanding concrete. They are polar, share the same
signed-permutation symmetry, and already behave very differently under the most immediate constructions. For
the cube,
\[
  \sum_{i=1}^d x_i^{2t}
  \qquad\text{gives}\qquad
  \dist_t(B_\infty^d)\le d^{1/(2t)},
\]
so logarithmic degree suffices for a constant certificate. For the cross-polytope, by contrast, Barvinok proved
an order-$\sqrt d$ sextic lower bound in the published version of his paper, and Barvinok and Veomett recorded
the quartic and sextic cases~\citep{Barvinok2003ApproxNorm}\citep[Sec.~2.3]{BarvinokVeomett2008}. These
fixed-degree obstructions suggest that the universal curve reflects a genuine worst case, rather than an
artifact of the lifted construction. The question is whether the obstruction persists for every degree, and
whether it depends on the SoS representation.

\paragraph{The degree law.}
Our main result answers both questions.

\begin{theorem}[Cross-polytope degree law; informal version of \cref{thm:main}]
\label{thm:main-informal}
For every $d\ge1$ and $1\le t\le d$,
\[
  \frac1{\sqrt{2e}}\sqrt{\frac dt}
  \le \distnn_t(B_1^d)
  \le \dist_t(B_1^d)
  \le \sqrt{2e}\sqrt{\frac dt}.
\]
\end{theorem}

Consequently, the cross-polytope matches Barvinok's universal SoS upper bound up to an absolute constant, and
constant distortion requires and is achieved by degree of order $d$. There are two nontrivial uniformities in
the statement. The estimate holds throughout $1\le t\le d$, rather than at a fixed degree, and its lower half
does not inspect an SoS representation at all. It extends the previously recorded fixed-degree cases to the
full degree range while retaining a representation-free lower bound over all globally nonnegative forms, and
hence over the full SoS cone. To the best of our knowledge, this is the first degree-uniform cross-polytope
lower bound over the full cone of globally nonnegative forms; the narrower sums-of-powers precedent and the
classical flat-ray literature are separated explicitly in the related-work discussion below.

The theorem asks for two logically different arguments. The lower bound must simplify \emph{every}
nonnegative form without making it harder to approximate the body. The upper bound needs only one form, but it
must explain the same $\sqrt{d/t}$ scale. Both arguments come from the signed-permutation symmetry, used in
opposite directions.

\paragraph{How the obstruction is isolated.}
A direct attack on the coefficients of an arbitrary nonnegative form retains far more information than the
sandwich ratio can detect. Signed-permutation averaging is the first compression: it erases orientation and
cannot increase the ratio. It does not, however, make the problem one-dimensional. An invariant degree-$2t$
form can still involve many products of even power sums.

The second compression is the decisive one. A lower bound does not require locating the true boundary extrema;
it is enough to find a small set of directions on which every candidate already varies. We choose the flat
points
\[
  x^{(k)}=\frac1k(\underbrace{1,\ldots,1}_{k},0,\ldots,0),
  \qquad k=1,\ldots,d,
\]
because on this family every even power sum depends only on the support size $k$. More importantly, a product
of power sums forgets the underlying partition and remembers only its number of factors. The restriction of
every invariant form therefore has the forced profile
\[
  V(k)=p(x^{(k)})=k^{-2t}Q(k),
  \qquad \deg Q\le t,\quad Q(0)=0.
\]
At this point the high-dimensional geometry has disappeared. Dimension supplies $d$ support scales at which
the profile must remain positive, while degree gives the numerator only $t$ degrees of freedom with which to
compensate for the factor $k^{-2t}$.

The mechanism is already visible when $t=1$: then $V(k)=a/k$, so the vertex and diagonal differ by a factor
$d$, yielding the classical $\sqrt d$ distortion. Higher degree can partially flatten this profile, but it
cannot tune the first $t$ support sizes and the distant scale $k=d$ independently. The change of variables
$u=1/k$ makes this precise: $V(1/u)=u^tR(u)$ with $\deg R\le t-1$, so the values at
$1,1/2,\ldots,1/t$ determine $R(1/d)$. Bounding that extrapolation functional by Lagrange interpolation gives
the lower bound. The candidate's strict positivity away from the origin keeps the sampled profile positive;
the argument never inspects how its global nonnegativity is certified.

\paragraph{How the matching form emerges.}
The upper construction is dictated by the same geometry rather than guessed as an unrelated probabilistic
device. The oriented facet normals of the cross-polytope are the sign vectors
$\varepsilon\in\{\pm1\}^d$; equivalently, opposite pairs define the slabs
$|\varepsilon^\top x|\le1$. Averaging the $2t$-th powers of these functionals gives the canonical symmetric form
\[
  f(x)=\E_\varepsilon[(\varepsilon^\top x)^{2t}].
\]
Schur convexity places its maximum on the $\ell_1$-sphere at a vertex and its minimum at the diagonal. Hence
\[
  \ratio(f^{1/(2t)})
  =\frac{d}{(\E[\Rad_d^{2t}])^{1/(2t)}},
  \qquad \Rad_d=\varepsilon_1+\cdots+\varepsilon_d.
\]
The collision-free pairings in this moment already supply the required scale, while Gaussian comparison gives
the opposite moment estimate. In the regime $t\to\infty$ with $t=o(\sqrt d)$, the displayed ratio is
$\sqrt{e/2}\sqrt{d/t}\,(1+o(1))$. We also retain Barvinok's upper bound because it is uniform beyond this
moment-asymptotic regime.

This gives the central proof map:
\[
\begin{array}{ccccc}
\text{arbitrary form}&\xrightarrow{\text{average}}&
\text{support-size profile}&\xrightarrow{\text{interpolate}}&
\text{lower bound},\\[2mm]
\text{one facet orbit}&\xrightarrow{\text{average}}&
\text{Rademacher form}&\xrightarrow{\text{moments}}&
\text{upper bound}.
\end{array}
\]
The two rows are deliberately asymmetric: the first must compress every adversary without increasing its
ratio, whereas the second needs to manufacture only one candidate with the matching order.

\paragraph{Polarity as a stress test.}
Is the lower-bound mechanism a generic consequence of polarity, or a feature of the $\ell_1$ geometry? The
cube certificate above and our cross-polytope lower bound settle the same-degree comparison without requiring
the cube's exact optimum. At $t=\lceil\log_2d\rceil$, the cube has an SoS certificate of distortion at most
$\sqrt2$, while every globally nonnegative certificate for the cross-polytope has distortion
\[
  \Omega\!\left(\sqrt{\frac d{\log d}}\right).
\]
Thus polarity does not preserve the degree scale at which a prescribed distortion is attained.

This comparison also provides an early scope check. The coordinate form for the cube is sufficient for the
separation, but sufficiency is not optimality. Already for $B_\infty^2$ at $t=2$, the SoS form
\[
  x^4+y^4-x^2y^2=(x^2-y^2)^2+(xy)^2
\]
has distortion $(4/3)^{1/4}<2^{1/4}$. The negative mixed coefficient is unavailable to a nonnegative mixture of
coordinate fourth powers, but it remains compatible with an SoS representation.

\paragraph{What the two proof mechanisms test next.}
The sign design suggests weighting powers of facet normals for a general symmetric polytope. Maximizing the
boundary floor gives a concave, one-sided certificate. In the oracle model, entropic mirror ascent gives an
additive outer-iteration bound, while homogeneity represents each worst-direction oracle by finitely many
convex dual-norm problems. At $t=1$ the program reduces to $G$-optimal design and the John
ellipsoid~\citep{KieferWolfowitz1960}. The quartic example
then supplies the necessary scope check: even a fully symmetric optimized facet design need not be optimal over
the full SoS cone.

The $\ell_p$ family instead tests the weighted flat-profile obstruction. Its exponent is $\alpha=2t/p$, so
$p=2$ is exactly the point where $\alpha$ meets the numerator degree. The resulting theorem gives the exact
quadratic law, exact algebraic resonance points, and matching dimension exponents up to $d^{o(1)}$ for fixed
$p<2$ with $\log t=o(\log d)$. Beyond the threshold, the argument stops without predicting the answer. The
formal statements of both extensions appear in \cref{sec:results}.

\paragraph{Related work and scope.}
The polynomial-approximation framework originates in \citet{Barvinok2003ApproxNorm} and the survey of
\citet{BarvinokVeomett2008}. Polynomial norms and the geometry separating nonnegative forms, sums of squares,
and sums of powers are studied in
\citet{AhmadiDeKlerkHall2019PolyNorms,Blekherman2006Volumes}. For the much narrower family of sums of
$2t$-th powers of linear forms, the polynomial root is an $L_{2t}$ subspace norm; the classical estimate
$T_2(L_{2t})=\Theta(\sqrt t)$ already gives an $\Omega(\sqrt{d/t})$ obstruction by testing coordinate vectors
and their Rademacher sums~\citep[Lem.~3.2]{DuembgenVanDeGeerVeraarWellner2010}. The scope gain here is that the
lower proof is representation-free and applies to every globally nonnegative form.

The support-size rays themselves are classical in symmetric nonnegativity. For $d\ge3$,
Choi--Lam--Reznick proved that an even symmetric sextic is nonnegative if and only if it is nonnegative on the
rays $(\mathbf 1_k,0^{d-k})$, $k=1,\ldots,d$~\citep[Thm.~3.7]{ChoiLamReznick1987EvenSymmetricSextics}.
Timofte placed this criterion in a broader theory of test sets with few distinct coordinate levels, followed by
half-degree and reflection-group developments
\citep{Timofte2003PositivitySymmetricI,Riener2012HalfDegree,AcevedoVelasco2016ReflectionTestSets}. Recent work
extends the flat-ray criterion to a hook-shaped subspace expressed in elementary symmetric polynomials, while
showing that the corresponding finite criterion for the power-sum hook-shaped class fails in every degree
$r\ge4$---equivalently, from even degree $2r\ge8$ after the square substitution
\citep[Prop.~5.5, Rem.~5.7, and Thms.~5.6, 5.8]{AcevedoBlekhermanDebusRiener2025Vandermonde}. Our use is
quantitative rather than a new test-set claim: we do not assert that these rays contain all extrema or test
nonnegativity in every degree. The normalization also matters. The customary simplex point is
$y=k^{-1}\mathbf 1_k$, whose pullback under $y_i=x_i^2$ is $x=k^{-1/2}\mathbf 1_k$; our $\ell_1$-boundary point
is instead $k^{-1}\mathbf 1_k$. Homogeneity gives
$p(k^{-1}\mathbf 1_k)=k^{-t}p(k^{-1/2}\mathbf 1_k)$. This radial weight is part of the profile whose all-degree
oscillation yields the distortion bound.

The reciprocal-grid step is adjacent to discrete norming and Remez-type estimates for polynomials bounded on
finite sets~\citep{CoppersmithRivlin1992Growth,Yomdin2011DiscreteRemez}. Its nodes, prescribed weight, and
objective are different: we control the oscillation of positive values of $k^{-\alpha}Q(k)$ on a reciprocal
grid and assume no positivity between grid points. We therefore use the estimate as a portable proof lemma,
not as a general replacement for discrete Remez theory.

The quadratic endpoint of the facet program is classical optimal design~\citep{KieferWolfowitz1960}; algorithms
for the associated John ellipsoid include \citet{CohenCousinsLeeYang2019} and subsequent refinements
\citep{LiYuJiangGaoHan2026BeyondAveraging,CaoLiSongYangZhou2022,WoodruffYasuda2025Lazy}. We use entropic mirror
ascent in its standard nonsmooth form~\citep{Bubeck2015ConvexOpt}.

\paragraph{Organization.}
\Cref{sec:prelim} fixes the sandwich framework. \Cref{sec:results} states the formal results in the order of the
questions above. \Cref{sec:lower} isolates the support-size obstruction, and \cref{sec:upper-alg} develops the
matching sign construction and then the general facet certificate, including its strict limitation.
\Cref{sec:discussion} digests what the argument proves, what it rules out, and which component must change in
the remaining problems. Deferred assemblies and the $\ell_p$ proofs appear in \cref{app:proofs}.

\section{Preliminaries}
\label{sec:prelim}

\paragraph{Basic notation.}
We write $[m]=\set{1,\ldots,m}$ and
$\simplex_n=\set{w\in\R^n:w_i\ge0,\ \sum_iw_i=1}$. When the ambient body is fixed and clear, we suppress it
from the subscript of the sandwich ratio $\ratio_K$.

\paragraph{Bodies and gauges.}
A body $K\subset\R^d$ is \emph{symmetric} if it is convex, compact, full-dimensional, and satisfies $K=-K$.
Its gauge is
\[
  \gauge{x}{K}=\inf\set{\lambda>0:x\in\lambda K}.
\]
For a symmetric polytope
$\Poly=\set{x\in\R^d:\abs{a_i^\top x}\le 1,\ i\in[n]}$, where each $a_i$ represents one opposite pair of
facet normals and the rows span $\R^d$, this becomes
$\gauge{x}{\Poly}=\max_i\abs{a_i^\top x}$.

\paragraph{Two polynomial cones.}
Let $\mc N_{d,2t}$ be the cone of globally nonnegative forms of degree $2t$, and let
\[
  \Sigma_{d,2t}
  =\set{p:p=\textstyle\sum_r q_r^2,\ \ q_r\text{ a form of degree }t}
\]
be the sum-of-squares (SoS) cone. In general $\Sigma_{d,2t}\subsetneq\mc N_{d,2t}$. This distinction matters:
Barvinok's construction lies in the SoS cone, whereas our lower bound for the cross-polytope holds against the
larger nonnegative cone.

For $p\in\mc N_{d,2t}$ that is positive away from the origin, put $g=p^{1/(2t)}$. The function $g$ is
positively homogeneous, but it need not be convex and need not be a norm. Its sandwich ratio against $K$ is
\[
  \ratio_K(g)
  =\frac{\max_{\gauge{x}{K}=1}g(x)}
         {\min_{\gauge{x}{K}=1}g(x)} .
\]
For a nonnegative form that vanishes at a nonzero point, we use the extended-value convention
$\ratio_K(p^{1/(2t)})=+\infty$. Thus degenerate forms can be included harmlessly in later infima over closed
weight simplices.
Writing $m=\min_{\partial K}g$ and $M=\max_{\partial K}g$, homogeneity gives
$m\set{g\le1}\subseteq K\subseteq M\set{g\le1}$, so the ratio $M/m$ is precisely the multiplicative loss.
We therefore define
\begin{equation}\label{eq:two-kappas}
  \dist_t(K)
  =\inf_{\substack{p\in\Sigma_{d,2t}\\ p(x)>0\text{ for }x\ne0}}\ratio_K(p^{1/(2t)}),
  \qquad
  \distnn_t(K)
  =\inf_{\substack{p\in\mc N_{d,2t}\\ p(x)>0\text{ for }x\ne0}}\ratio_K(p^{1/(2t)}).
\end{equation}
Thus $\distnn_t(K)\le\dist_t(K)$. Unless a superscript is shown, $\dist_t$ denotes the SoS distortion.
At $t=1$ the two cones coincide and $\dist_1(K)$ is the best ellipsoidal sandwich factor; John's theorem gives
$\dist_1(K)\le\sqrt d$.

\paragraph{Exact degree and cumulative degree budgets.}
The index $t$ in \eqref{eq:two-kappas} specifies the exact homogeneous degree $2t$. Successive feasible classes
are not literally nested, so no monotonicity in $t$ is assumed. If a monotone degree budget is desired, define
\[
  \widehat\kappa_t(K)=\min_{1\le s\le t}\dist_s(K),
  \qquad
  \widehat\kappa_t^{\mathrm{nn}}(K)=\min_{1\le s\le t}\distnn_s(K).
\]
For the cross-polytope, \cref{thm:main} gives the same $\Theta(\sqrt{d/t})$ order for these envelopes when
$1\le t\le d$: use the theorem's lower bound at every $s\le t$ and its upper bound at $s=t$.

\paragraph{Facet-power certificates.}
If $\Poly$ is a symmetric polytope and $w\in\simplex_n$, then
\[
  f_{t,w}(x)=\sum_i w_i(a_i^\top x)^{2t}
            =\sum_i w_i\big((a_i^\top x)^t\big)^2
\]
is an SoS form. On $\partial\Poly$ it satisfies
$f_{t,w}(x)\le\max_i(a_i^\top x)^{2t}=1$. We use these forms for upper bounds and for the oracle design problem
of \cref{sec:upper-alg}; they are a strict subclass of all SoS forms.
The three levels used in the paper are therefore
\[
  \mc N_{d,2t}
  \ \supseteq\ 
  \Sigma_{d,2t}
  \ \supseteq\ 
  \set{f_{t,w}:w\in\simplex_n},
\]
with the inclusion directions reversed for minimization: enlarging the candidate class can only lower the best
distortion. The main lower bound controls the first and largest class; the facet program optimizes a certificate
inside the last and smallest one.

\paragraph{The Veronese lift.}
Let $\ver_t:\R^d\to\R^{\Dt}$, $\Dt=\binom{d+t-1}{t}$, be the symmetric degree-$t$ Veronese embedding, normalized
so that $\inner{\ver_t(a)}{\ver_t(x)}=(a^\top x)^t$ (the polynomial-kernel feature map). Then
$(a_i^\top x)^{2t}=\inner{\ver_t(a_i)}{\ver_t(x)}^2$ and
\[
  f_{t,w}(x)\ =\ \ver_t(x)^\top \Mt(w)\,\ver_t(x),\qquad \Mt(w)=\sum_i w_i\,\ver_t(a_i)\ver_t(a_i)^\top ,
\]
so the sublevel set of a facet-power form is the inverse image, under $\ver_t$, of a possibly degenerate
ellipsoid in the lifted space $\R^{\Dt}$. For positive weights---or, more generally, spanning support---the
form is positive away from the origin. This is the sense in which each degree level lifts the John ellipsoid: at
$t=1$, $\ver_1$ is the identity and $\Mt(w)=\sum_i w_i a_i a_i^\top$ is the usual information matrix. We use the
lift for the algorithm of \cref{sec:upper-alg}; the lower bound of \cref{sec:lower} works directly in $\R^d$.

\paragraph{Barvinok's baseline.}
Barvinok's construction, in the all-degree formulation recorded by Barvinok and Veomett, gives for every
symmetric body and every $t\ge1$ an SoS form satisfying~\citep{Barvinok2003ApproxNorm,BarvinokVeomett2008}
\begin{equation}\label{eq:barvinok-prelim}
  \dist_t(K)\le\binom{d+t-1}{t}^{1/(2t)}.
\end{equation}
For $1\le t\le d$, the right-hand side is at most $\sqrt{2ed/t}$.
Barvinok and Veomett also recorded the following explicit construction for the
cube~\citep[Sec.~2.3]{BarvinokVeomett2008}:
\begin{fact}[A logarithmic-degree cube certificate]
\label{fact:cube-upper}
  The SoS form $p(x)=\sum_{i=1}^d x_i^{2t}$ satisfies
  \[
    \dist_t(\Binf)\le d^{1/(2t)}.
  \]
  In particular, $t=\lceil\log_2d\rceil$ gives $\dist_t(\Binf)\le\sqrt2$.
\end{fact}
\noindent The displayed quantity is an upper bound, not the exact full-cone distortion; \cref{prop:design-gap}
gives a strict SoS improvement already for $d=t=2$. By contrast, Barvinok and Veomett proved that the
cross-polytope still has order-$\sqrt d$ distortion at degrees $4$ and $6$. Our main theorem extends that lower
bound to every $1\le t\le d$.

\paragraph{Rademacher moments.}
Let $\Rad_d=\sum_{j=1}^d\varepsilon_j$ be a sum of independent Rademacher signs $\varepsilon_j\in\set{\pm 1}$, and
$(2t-1)!!=\frac{(2t)!}{2^t t!}=\prod_{i=1}^t(2i-1)$ the double factorial (the $2t$-th Gaussian moment). We write
$B_d$ for the hyperoctahedral group of signed permutations of $\R^d$, the symmetry group of both $\Binf$ and
$\Bone$. The $k$-flat points are $\kflat{k}=\tfrac1k(\underbrace{1,\dots,1}_{k},0,\dots,0)$; each lies on the
$\ell_1$-sphere, $\gauge{\kflat{k}}{1}=1$.

\section{Main results}
\label{sec:results}

The formal results follow the same diagnostic chain as the introduction. What forces the universal rate? Which
explicit form matches it? Does the obstruction survive polarity? How much of the upper construction extends to
general facet systems, and where does that restricted model fail? Finally, which part of the mechanism remains
visible for $\ell_p$ balls? Proofs appear in \cref{sec:lower,sec:upper-alg,app:proofs}.

\subsection{The extremal cross-polytope}

\begin{restatable}[Cross-polytope degree law]{theorem}{ThmMain}
\label{thm:main}
For all $d\ge 1$ and $1\le t\le d$,
\begin{equation}\label{eq:main-bound}
  \frac{1}{\sqrt{2e}}\sqrt{\frac dt}
  \ \le\ \distnn_t(\Bone)
  \ \le\ \dist_t(\Bone)
  \ \le\
  \min\left\{
    \binom{d+t-1}{t}^{1/(2t)},
    \frac{d}{\big(\E[\Rad_d^{2t}]\big)^{1/(2t)}}
  \right\}
  \ \le\ \sqrt{2e}\sqrt{\frac dt}.
\end{equation}
In particular, both the SoS and the nonnegative-form distortions are $\Theta(\sqrt{d/t})$. Moreover, the sign
design has the sharper asymptotic value
\[
  \frac{d}{\big(\E[\Rad_d^{2t}]\big)^{1/(2t)}}
  =\sqrt{\frac e2}\sqrt{\frac dt}\,(1+o(1))
  \qquad\text{when }t\to\infty\text{ and }t=o(\sqrt d).
\]
Consequently:
\begin{enumerate}
  \item\label{it:extremal} the cross-polytope matches Barvinok's universal upper bound up to an absolute
  constant, and is therefore extremal among symmetric bodies up to constants;
  \item\label{it:collapse} for every fixed $C>1$, degree $2t=\Theta_C(d)$ is necessary and sufficient for either
  distortion to be at most $C$.
\end{enumerate}
\end{restatable}

\begin{proof}[Proof idea]
Average an arbitrary nonnegative form over signed permutations and restrict it to the flat points; invariant
theory forces the support-size profile $V(k)=k^{-2t}Q(k)$ with $\deg Q\le t$ and $Q(0)=0$, and interpolation
forces its variation. For the other direction, uniformly averaging the sign-vector facets gives the
Rademacher design of \cref{lem:upper}, while Barvinok's theorem supplies the remaining uniform upper bound.
\end{proof}

\paragraph{Why the rate has this form.}
At $t=1$, an invariant quadratic is a multiple of $\sum_i x_i^2$, whose value at $\kflat{k}$ is proportional
to $1/k$; the vertex and diagonal already differ by a factor $d$, producing the classical $\sqrt d$ distortion.
Higher degree does not remove this tension. It gives the numerator degree $t$ with which to flatten a profile
whose geometric normalization is $k^{2t}$. Interpolation proves that this extra budget buys the factor
$\sqrt t$, but cannot buy more.

Barvinok proved the cross-polytope lower bound for $t=3$, and Barvinok and Veomett recorded the cases
$t=2,3$~\citep{Barvinok2003ApproxNorm}\citep[Sec.~2.3]{BarvinokVeomett2008}. The lower half of
\eqref{eq:main-bound}, proved in \cref{sec:lower}, extends these fixed-degree precedents to every
$1\le t\le d$ while retaining a comparison against the full cone of nonnegative forms.

\begin{remark}[The nontrivial regime]\label{rem:regime}
The displayed lower bound exceeds the trivial bound $1$ when $t<d/(2e)$. For larger $t\le d$, both
$\sqrt{d/t}$ and the distortion are of constant order. The necessity statement in
part~\ref{it:collapse} of \cref{thm:main} uses the lower bound at $t=\Theta_C(d)$; sufficiency follows from Barvinok's
bound for all $t$, including $t>d$ when $C$ is close to $1$.
\end{remark}

The upper construction is forced by the facet geometry. The oriented facet normals of $\Bone$ are sign
vectors, so the most symmetric facet-power form is their uniform average.

\begin{restatable}[The sign design]{lemma}{LemUpper}
\label{lem:upper}
For $\Bone$ and $1\le t\le d$, the uniform design over the sign vectors $\set{\pm1}^d$ gives the form
$f(x)=\E_{\varepsilon\sim\set{\pm1}^d}[(\varepsilon^\top x)^{2t}]$ with
$\ratio(f^{1/(2t)})=d/(\E[\Rad_d^{2t}])^{1/(2t)}$, and
\[
  \frac{\sqrt d}{\big((2t-1)!!\big)^{1/(2t)}}
  \ \le\ \ratio(f^{1/(2t)})
  \ \le\ e\sqrt{\frac dt}.
\]
If $\binom t2<d$, the upper bound can be sharpened to
\[
  \ratio(f^{1/(2t)})
  \le
  \frac{\sqrt d}
  {\big((2t-1)!!\big)^{1/(2t)}
   \big(1-\binom{t}{2}/d\big)^{1/(2t)}} .
\]
Thus $\dist_t(\Bone)\le\ratio(f^{1/(2t)})$, and the ratio is
$\sqrt{e/2}\sqrt{d/t}\,(1+o(1))$ when $t\to\infty$ and $t=o(\sqrt d)$.
\end{restatable}

\begin{proof}[Proof idea]
The form $f$ is invariant under signed permutations. Equalizing two coordinates while preserving their
$\ell_1$ mass can only lower $f$; equivalently, Schur convexity pins the maximum, $1$, at a vertex $e_1$ and the
minimum, $d^{-2t}\E[\Rad_d^{2t}]$, at the diagonal $\kflat{d}$. Geometry is thereby reduced to one Rademacher
moment. Gaussian comparison bounds this moment from above, while the collision-free perfect-matching terms in
its expansion already give the correct lower scale. See \cref{sec:upper-alg}.
\end{proof}

\subsection{The lower bound, in one variable}

The main new ingredient is the lower bound. It uses two distinct compressions. Symmetrization replaces an
arbitrary form by an invariant one without improving the candidate's ratio. Flat-point evaluation then replaces
the remaining invariant algebra by a support-size curve. Only the second step makes the problem one-dimensional.

\begin{restatable}[Symmetrization and flat-point reduction]{lemma}{LemReduction}
\label{lem:reduction}
Let $p$ be any nonnegative form of degree $2t$, positive on the $\ell_1$-sphere. Then:
\begin{enumerate}
  \item averaging $p$ over $B_d$ preserves nonnegativity and does not increase its sandwich ratio; if $p$ is
  SoS, its average is SoS as well. Consequently both infima in \eqref{eq:two-kappas} may be restricted to
  $B_d$-invariant forms;
  \item a $B_d$-invariant form is a linear combination of products of even power sums, and its values on the
  flat points are
  \[
    V(k)\ \defeq\ p(\kflat{k})\ =\ \sum_{\ell=1}^{t} b_\ell\, k^{\ell-2t}\ =\ k^{-2t}\,Q(k),\qquad
    Q(k)=\sum_{\ell=1}^t b_\ell k^\ell ,
  \]
  with $Q$ a real polynomial of degree $\le t$ and $Q(0)=0$;
  \item consequently $\ratio(p^{1/(2t)})^{2t}\ \ge\ \max_{k\in[d]}V(k)\big/\min_{k\in[d]}V(k)$.
\end{enumerate}
\end{restatable}

After this reduction, all geometry used by the proof is contained in the positive numbers
$V(1),\ldots,V(d)$. The form away from the flat points, and any certificate representing its nonnegativity,
have disappeared.

The remaining estimate is cleaner when the prescribed power of $k$ is left as a parameter. This separates the
geometric input, which determines $\alpha$, from the interpolation step, which sees only the gap between
$\alpha$ and the available polynomial degree.

\begin{restatable}[Weighted reciprocal-grid profile lemma]{lemma}{LemWeightedProfile}
\label{lem:weighted-profile}
Let $d,t$ be positive integers with $1\le t\le d$, let $\alpha\in\R$ satisfy $\alpha\ge t$, and let $Q$ be a
real polynomial with $\deg Q\le t$ and $Q(0)=0$. If $V(k)=k^{-\alpha}Q(k)>0$ for every $k\in[d]$, put
$\beta=\alpha-t$ and let $\ell_i$ be the Lagrange basis polynomial for the reciprocal nodes
$1,1/2,\ldots,1/t$. Then
\begin{equation}\label{eq:weighted-profile}
  \frac{\max_{k\in[d]}V(k)}{\min_{k\in[d]}V(k)}
  \ \ge\
  \max\left\{1,\frac{d^\beta}{\sum_{i=1}^t i^\beta\abs{\ell_i(1/d)}}\right\}
  \ \ge\
  \max\left\{1,2(2e)^{-t}\left(\frac dt\right)^\beta\right\}.
\end{equation}
For $t=1$, the exact ratio is $d^{\alpha-1}$.
\end{restatable}

The exponent gap in \eqref{eq:weighted-profile} is structural. At $\alpha=t$, the choice $Q(k)=k^t$ gives a
constant profile, so no dimension growth is possible without additional information. When $\alpha>t$, the
known weight creates the possibility of dimension growth, but the conclusion becomes nontrivial only when the
displayed weighted interpolation term exceeds one. On a short grid an exactly level profile may still exist;
this is why the maximum with $1$ in \eqref{eq:weighted-profile} is essential.

\begin{restatable}[A discrete univariate bound]{lemma}{LemUnivariate}
\label{lem:univariate}
There is an absolute constant $c_1=1/\sqrt{2e}$ such that for every $d\ge 1$, every $1\le t\le d$, and every real
polynomial $Q$ of degree $\le t$ with $Q(0)=0$ and $Q(k)>0$ for all $k\in[d]$, the sequence $V(k)=k^{-2t}Q(k)$
satisfies
\[
  \frac{\max_{k\in[d]}V(k)}{\min_{k\in[d]}V(k)}
  \ \ge\ 2\left(\frac{d}{2et}\right)^t
  \ \ge\ \big(c_1\sqrt{d/t}\big)^{2t}.
\]
For $t=1$, the proof recovers the exact ratio $d$ before the final factorial relaxation. The displayed lower
bound is at least $2$ whenever $t\le d/(2e)$.
\end{restatable}

Together, \cref{lem:reduction,lem:univariate} give the stronger bound
$\distnn_t(\Bone)\ge c_1\sqrt{d/t}$. The weighted version also removes a repeated interpolation calculation
from the $\ell_p$ proof. We sketch its argument here and give the full proof in \cref{sec:lower}.

\begin{proof}[Proof idea for \cref{lem:weighted-profile}]
Put $\beta=\alpha-t\ge0$ and write $V(1/u)=u^\beta R(u)$, where $R$ has degree at most $t-1$. Its values at
$u=1,1/2,\ldots,1/t$ determine $R(1/d)$ exactly. The data satisfy
$R(1/i)=i^\beta V(i)$, while $V(d)=d^{-\beta}R(1/d)$; if all $V(k)$ were nearly equal, the factor
$(t/d)^\beta$ would force the last value to be too small. The reciprocal substitution thus exposes both the
effective degree and the exponent gap. Bounding the Lagrange weights and using $t!\ge(t/e)^t$ gives
\eqref{eq:weighted-profile}; setting $\alpha=2t$ gives \cref{lem:univariate}.
\end{proof}

\subsection{A same-degree contrast with the cube}

Polarity is a stress test for the obstruction. The cube construction of \cref{fact:cube-upper} and the
cross-polytope lower bound can be compared at the same degree without knowing the cube's exact optimum.

\begin{restatable}[Polar bodies at the same degree]{corollary}{CorSeparation}
\label{cor:separation}
Let $t_C^*(K)=\min\set{t\ge1:\dist_t(K)\le C}$. For every fixed $C>1$ and every $d\ge2$,
\[
  t_C^*(\Binf)\ \le\ \left\lceil\frac{\log d}{2\log C}\right\rceil,
  \qquad
  t_C^*(\Bone)=\Theta_C(d).
\]
In particular, for $d\ge2$ and $t=\lceil\log_2d\rceil$,
\[
  \dist_t(\Binf)\le\sqrt2,
  \qquad
  \dist_t(\Bone)\ge
  \frac1{\sqrt{2e}}\sqrt{\frac{d}{\lceil\log_2d\rceil}}.
\]
Thus an explicit logarithmic-degree SoS form already gives constant distortion for the cube, while the polar
cross-polytope still has diverging distortion and requires degree of order $d$ for a constant factor.
\end{restatable}

\subsection{An oracle max--min design}

Facet-power forms lead to a useful certificate problem. Its logic has four steps: the boundary floor is concave
in the facet weights; a worst boundary point is a supergradient; homogeneity converts the nonconvex-looking
boundary search into dual-norm optimization; and at $t=1$ the entire picture collapses to classical optimal
design. The result below states exactly what this optimization certifies.

\begin{restatable}[Optimization of a facet-design bound]{theorem}{ThmAlgorithm}
\label{thm:algorithm}
For a symmetric polytope
$\Poly=\set{x:\abs{a_i^\top x}\le 1,\ i\in[n]}$, whose facet normals span $\R^d$, let
$\Phi_t(w)=\min_{\gauge{x}{\Poly}=1}f_{t,w}(x)$
and define the facet-design certificate
\[
  \designbd_t(\Poly)=\left(\max_{w\in\simplex_n}\Phi_t(w)\right)^{-1/(2t)}.
\]
The maximum is attained and is strictly positive, so this certificate is finite.
Then:
\begin{enumerate}
  \item\label{it:concave} $\Phi_t$ is concave and takes values in $[0,1]$. Moreover,
  \[
    \distnn_t(\Poly)\le\dist_t(\Poly)
    \le\inf_{w\in\simplex_n}\ratio_\Poly(f_{t,w}^{1/(2t)})
    \le\designbd_t(\Poly).
  \]
  \item\label{it:mwu} For $n\ge2$, let $T\ge1$ and $\delta\ge0$. Suppose an oracle, given $w$, returns
  $x_w\in\partial\Poly$ with
  $f_{t,w}(x_w)\le\Phi_t(w)+\delta$. Entropic mirror ascent with
  \[
    w_i^+
    =\frac{w_i\exp\!\big(\eta(a_i^\top x_w)^{2t}\big)}
           {\sum_jw_j\exp\!\big(\eta(a_j^\top x_w)^{2t}\big)}
  \]
  initialized at the uniform distribution, with $\eta=\sqrt{2\log n/T}$, has an averaged iterate $\bar w_T$
  satisfying
  \[
    \max_{w\in\simplex_n}\Phi_t(w)-\Phi_t(\bar w_T)
    \le\sqrt{\frac{2\log n}{T}}+\delta.
  \]
  Thus, for every $\eps>0$, an exact oracle gives additive $\eps$-accuracy in
  $T=O(\max\{1,\log n\}/\eps^2)$ calls. The case $n=1$ is immediate.
  \item\label{it:oracle} For $w$ in the relative interior of the simplex, put
  $h_w=f_{t,w}^{1/(2t)}$. Then $h_w$ is a norm and
  \[
    \Phi_t(w)^{-1/(2t)}
    =\max_{i\in[n]}h_w^*(a_i),
    \qquad
    h_w^*(a_i)=\max\set{a_i^\top x:f_{t,w}(x)\le1}.
  \]
  Thus the worst-direction oracle reduces exactly to $n$ convex norm-dual problems.
  \item\label{it:ccly} At $t=1$, for $M(w)=\sum_iw_i a_i a_i^\top\succ0$,
  \[
    \Phi_1(w)=\frac{1}{\max_i a_i^\top M(w)^{-1}a_i}.
  \]
  Hence $\max_w\Phi_1(w)=1/d$ and its maximizers are the $D$-optimal designs by the Kiefer--Wolfowitz equivalence
  theorem. For any maximizer $w^\star$, the ellipsoid
  $\set{x:d\,x^\top M(w^\star)x\le1}$ is the symmetric John ellipsoid.
\end{enumerate}
\end{restatable}

\paragraph{Proof mechanism.}
The ceiling $\max_{\partial\Poly}f_{t,w}\le1$ turns certificate design into the problem of raising a concave
boundary floor. A minimizing boundary point exposes a supergradient, so entropic mirror ascent repeatedly
reweights the facet directions that see the least-protected point. Homogeneity then replaces the boundary
equality constraint by dual norms on the unit ball of $h_w$. At $t=1$, these dual norms are leverage scores and
the construction reduces to classical $D/G$-optimal design.

\paragraph{What the theorem does not optimize.}
The theorem optimizes the certificate $\designbd_t$; it neither minimizes the true ratio within the facet-power
family nor characterizes the full SoS optimum. The norm-dual identity removes the apparent nonconvexity of the
boundary oracle, but its numerical cost and conditioning remain separate from the outer iteration bound.

Before treating the optimized certificate as a proxy for the SoS optimum, one should test the smallest
symmetric nonquadratic case. A mixed quartic term already separates the two notions.

\begin{restatable}[A symmetric design gap]{proposition}{PropDesignGap}
\label{prop:design-gap}
For the square and quartic forms,
\[
  \dist_2(B_\infty^2)
  \le\left(\frac43\right)^{1/4}
  <2^{1/4}
  =\designbd_2(B_\infty^2).
\]
The same strict gap holds for $B_1^2$, which is linearly equivalent to $B_\infty^2$. Thus signed-permutation
symmetry alone does not make the facet-power certificate optimal over the full SoS cone.
\end{restatable}

\subsection{\texorpdfstring{Extensions to $\ell_p$ balls}{Extensions to ell-p balls}}
\label{sec:ellp}

The $\ell_p$ family stress-tests the flat-point mechanism by changing one exponent while preserving the
symmetry. It reveals a methodological transition at $p=2$. We state only the ranges supported by the proof;
determining the interior regime $p>2$ remains open.

\begin{restatable}[Three consequences for $\ell_p$ balls]{theorem}{ThmEllp}\label{thm:ellp}
For $1\le t\le d$ and $p\in[1,\infty]$, with the convention $1/\infty=0$:
\begin{enumerate}
  \item\label{it:reson} If $p=2j$ for a divisor $j\mid t$, then
  $\distnn_t(\Bp)=\dist_t(\Bp)=1$.
  \item\label{it:t1law} At degree two,
  $\dist_1(\Bp)=d^{\abs{1/2-1/p}}$.
  \item\label{it:plaw} If $1\le p\le2$, then
  \[
    \frac1{\sqrt{2e}}\left(\frac dt\right)^{1/p-1/2}
    \le\distnn_t(\Bp)
    \le\dist_t(\Bp)
    \le d^{1/p-1/2}.
  \]
  Consequently, for fixed $p<2$ and $\log t=o(\log d)$, both distortions equal
  $d^{\,1/p-1/2+o(1)}$.
\end{enumerate}
\end{restatable}

The resonance statement in \cref{thm:ellp} occurs because
$\gauge{x}{2j}^{2t}=(\sum_i x_i^{2j})^{t/j}$ is itself an SoS form of degree $2t$. For
$1\le p\le2$, evaluating an invariant form at
$k^{-1/p}(1,\ldots,1,0,\ldots,0)$ replaces the exponent $2t$ in
\cref{lem:reduction} by $2t/p$, after which \cref{lem:weighted-profile} applies directly. For $p>2$ the exponent
crosses the interpolation degree, and this proof no longer determines the distortion. This is a boundary of
the argument, not evidence for any particular law beyond it.

\section{The obstruction: from a boundary problem to a grid}
\label{sec:lower}

This section proves the stronger inequality
$\distnn_t(\Bone)\ge c_1\sqrt{d/t}$, where $c_1=1/\sqrt{2e}$. The proof performs three reductions, each
discarding information in a direction safe for a lower bound. We first remove orientation by averaging, then
choose an orbit family on which the invariant algebra remembers only support size, and finally invert that
support-size grid so that ordinary polynomial interpolation applies. The decisive quantitative step is to use
the classical flat-ray family as a distortion witness. Classifying invariant nonnegative forms directly would
still leave many partition-indexed terms; on these rays all of those terms collapse onto support size. Once the profile
$k^{-2t}Q(k)$ appears, the remaining obstruction is one-dimensional.

\subsection{Choosing the witness profile}

\LemReduction*

\begin{proof}
Write $S=\set{x:\gauge{x}{1}=1}$ for the $\ell_1$-sphere and $G=B_d$ for the group of signed permutations. Since
signed permutations preserve $\ell_1$, the set $S$ is $G$-invariant.

\emph{(1) Symmetrization.} Let $p$ be nonnegative of degree $2t$, positive on $S$. Define the orbit average
$\bar p(x)=\tfrac1{\abs G}\sum_{U\in G}p(U^{-1}x)$; it is nonnegative, of degree $2t$, and $G$-invariant. For each
$x\in S$ every point $U^{-1}x$ lies in $S$, so $\min_S p\le p(U^{-1}x)\le\max_S p$, and averaging preserves the
bracket: $\min_S p\le\bar p(x)\le\max_S p$ for all $x\in S$. Taking extrema over $x$ gives $\max_S\bar p\le\max_S
    p$ and $\min_S\bar p\ge\min_S p$, hence
    $\ratio(\bar p^{1/(2t)})\le\ratio(p^{1/(2t)})$. If
    $p=\sum_r q_r^2$ is SoS, then
    $\bar p=\abs{G}^{-1}\sum_{U,r}(q_r\circ U^{-1})^2$ is SoS as well.
    It follows that the infimum over either polynomial cone is unchanged when restricted to $G$-invariant forms.

    This bracket argument is the entire reason averaging is safe: orbit averaging may pull a maximum down and
    push a minimum up, but it cannot increase the ratio. This is exactly the direction needed when minimizing
    over candidates. Notice also that no positivity representation has been used.

\emph{(2) Power-sum coordinates.} A $G$-invariant polynomial is invariant under sign flips --- so each variable
appears to even powers, i.e.\ it is a polynomial in $y_i=x_i^2$ --- and under permutations, so it is a symmetric
polynomial in $y_1,\dots,y_d$. By the fundamental theorem of symmetric polynomials and Newton identities, these
are generated by the power sums $\sum_i y_i^k=\sum_i x_i^{2k}=p_{2k}(x)$. A degree-$2t$ form is therefore a combination of monomials
$\prod_k p_{2k}^{m_k}$ with $\sum_k k\,m_k=t$. On a flat point, $p_{2k}(\kflat{j})=j\cdot(1/j)^{2k}=j^{1-2k}$, so
a monomial with $\ell=\sum_k m_k$ factors evaluates to $j^{\,\ell-2t}$ (using $\sum_k k m_k=t$), where
$1\le\ell\le t$ ($\ell=1$ for the single factor $p_{2t}$, $\ell=t$ for $p_2^t$). Aggregating by $\ell$,
\[
  V(k)=p(\kflat{k})=\sum_{\ell=1}^t b_\ell\,k^{\ell-2t}=k^{-2t}Q(k),\qquad Q(k)=\sum_{\ell=1}^t b_\ell k^\ell,
\]
so $\deg Q\le t$ and $Q(0)=0$.

This is the information collapse on which the proof turns. At $\kflat{j}$, each power-sum factor contributes
one copy of $j$, whereas the weighted degree contributes the common factor $j^{-2t}$. The shape of the
partition indexing the invariant monomial disappears; only its number $\ell$ of factors remains. A large
invariant coefficient space has therefore become one polynomial of degree at most $t$.

\emph{(3) Flat points bound the ratio.} The $d$ flat points lie on $S$, so $\max_S p\ge\max_{k\in[d]}V(k)$ and
$\min_S p\le\min_{k\in[d]}V(k)$; dividing, $\ratio(p^{1/(2t)})^{2t}=\max_S p/\min_S p\ge\max_k V(k)/\min_k V(k)$.
\end{proof}

\begin{remark}
Part (3) is a \emph{lower} bound on $\ratio$: it does not claim the worst direction of a general invariant form is
a flat point. This is precisely why the witness strategy is economical. Locating the true extrema could require
understanding the form on the entire boundary; the proof needs only $d$ directions that already force
substantial variation.
\end{remark}

\subsection{Weighted grid rigidity by interpolation}

\begin{proof}[Proof of \cref{lem:weighted-profile}]
Fix $d,t,\alpha$ and a valid $Q$. Put $\beta=\alpha-t\ge0$, and write
$s=\max_{k\in[d]}V(k)$ and $m=\min_{k\in[d]}V(k)$, both positive.

The goal is to use the first $t$ support sizes to control $V(d)$. In its original variable, $V(k)$ is a
weighted polynomial profile. Passing to reciprocals separates the known power $u^\beta$ from the low-degree
part that is free to move. This is also why $\beta=\alpha-t$ is the relevant parameter rather than $\alpha$
alone.

\emph{Change of variable.} Since $Q(0)=0$ and $\deg Q\le t$, write $Q(k)=\sum_{j=1}^t a_j k^j$, so
$V(k)=\sum_{j=1}^t a_j k^{j-\alpha}$. With $u=1/k$,
\[
  V(1/u)=u^\beta R(u),
  \qquad
  R(u)=\sum_{j=1}^t a_j u^{t-j},
  \qquad \deg R\le t-1.
\]
Thus $R(1/i)=i^\beta V(i)$ and $V(d)=d^{-\beta}R(1/d)$.

\emph{Interpolate at the $t$ largest grid points.} Let $u_i=1/i$ for $i\in[t]$ and put $u_0=1/d$. Since
$\deg R\le t-1$, Lagrange interpolation through $u_1,\dots,u_t$ is exact at $u_0$:
\[
  R(u_0)=\sum_{i=1}^t R(u_i)\,\ell_i(u_0),\qquad \ell_i(u_0)=\prod_{\substack{1\le j\le t\\ j\ne i}}
  \frac{u_0-u_j}{u_i-u_j}.
\]
Put $L_i=\abs{\ell_i(u_0)}$. Since $R(u_i)=i^\beta V(i)$ and $V(i)\le s$,
\begin{equation}\label{eq:Vd-bound}
  V(d)=d^{-\beta}\abs{R(u_0)}
  \le d^{-\beta}s\sum_{i=1}^t i^\beta L_i.
\end{equation}
When $t=d$, $u_0=u_t$, so $L_t=1$ and the other weights vanish; the formula simply reads $V(d)\le s$.

\emph{Bound the interpolation weights.} With $u_0=1/d$ and $u_j=1/j$ ($1\le j\le t\le d$),
\[
  \abs{u_0-u_j}=\frac{d-j}{dj}\le\frac1j,\qquad \abs{u_i-u_j}=\frac{\abs{j-i}}{ij}.
\]
Hence $\prod_{j\ne i}\abs{u_0-u_j}\le\prod_{j\ne i}\tfrac1j=\tfrac{i}{t!}$ and
$\prod_{j\ne i}\abs{u_i-u_j}=\tfrac{(i-1)!\,(t-i)!}{i^{\,t-2}\,t!}$, so
\[
  L_i\le\frac{i/t!}{(i-1)!\,(t-i)!/(i^{\,t-2}t!)}=\frac{i^{\,t-1}}{(i-1)!\,(t-i)!}
  =\frac{i^{\,t-1}}{(t-1)!}\binom{t-1}{i-1}.
\]
Summing, with $i^{\,t-1}\le t^{\,t-1}$ and $\sum_{i=1}^t\binom{t-1}{i-1}=2^{\,t-1}$,
\begin{equation}\label{eq:sumL}
  \sum_{i=1}^t L_i\ \le\ \frac{t^{\,t-1}}{(t-1)!}\cdot 2^{\,t-1}.
\end{equation}

The weighted evaluation norm $\sum_i i^\beta L_i$ gives the first bound in
\eqref{eq:weighted-profile}. The simpler constant uses only its $\ell_1$ relaxation; the signs of the
Lagrange weights and the exact extremizing polynomial are irrelevant for that estimate.

\emph{Assemble.} Since $m\le V(d)$ and $s/m\ge1$, \eqref{eq:Vd-bound} first gives
\[
  \frac{s}{m}
  \ \ge\
  \max\left\{1,\frac{d^\beta}{\sum_i i^\beta L_i}\right\}.
\]
Because $i^\beta\le t^\beta$, \eqref{eq:sumL} and
$t!\ge(t/e)^t$ (from $e^t\ge t^t/t!$) give $(t-1)!\ge t^{\,t-1}/e^t$ and therefore
\[
  \frac{s}{m}
  \ \ge\
  \max\left\{1,
  \frac{(t-1)!}{2^{t-1}t^{t-1}}\left(\frac dt\right)^\beta\right\}
  \ \ge\
  \max\left\{1,2(2e)^{-t}\left(\frac dt\right)^\beta\right\}.
\]
At $t=1$, positivity forces $Q(k)=a_1k$ with $a_1>0$, so $V(k)=a_1k^{1-\alpha}$ and the exact ratio is
$d^{\alpha-1}$.

The estimate should be read backwards as follows. If all of the first $t$ profile values stay below $s$, they
determine the low-degree part $R$ and force the diagonal value $V(d)$ to pay the factor $(t/d)^\beta$, up to
the weighted interpolation norm. Once that weighted term exceeds one, some sampled values must separate. The
factorials quantify the cost of extrapolating from the first $t$ reciprocal scales to the last one. At
$\beta=0$, the example $Q(k)=k^t$ shows why the method correctly produces no dimension growth; for $\beta>0$
on a short grid, the maximum with one still cannot be discarded.
\end{proof}

\begin{proof}[Proof of \cref{lem:univariate}]
Specialize the preceding calculation to $\alpha=2t$, hence $\beta=t$. Dropping the maximum with $1$ gives
\[
  \frac{\max_kV(k)}{\min_kV(k)}
  \ge 2(2e)^{-t}\left(\frac dt\right)^t
  =2\left(\frac{d}{2et}\right)^t.
\]
Dropping the leading factor $2$ gives
$\big((2e)^{-1/2}\sqrt{d/t}\big)^{2t}$. At $t=1$, the exact statement in
\cref{lem:weighted-profile} gives ratio $d$; and the displayed coarse bound is at least $2$ whenever
$t\le d/(2e)$.
\end{proof}

\begin{proof}[Proof of the lower bound in \cref{thm:main}]
Take any nonnegative degree-$2t$ form that is positive away from the origin. By \cref{lem:reduction}, after
averaging we obtain a polynomial $Q$ with degree at most $t$, with $Q(0)=0$ and $Q(k)>0$ on $[d]$, such that the
form's ratio to the power $2t$ is at least
$\max_k k^{-2t}Q(k)/\min_k k^{-2t}Q(k)$. By \cref{lem:univariate}, this is at least
$(c_1\sqrt{d/t})^{2t}$. Taking the infimum over all nonnegative forms gives
$\distnn_t(\Bone)\ge c_1\sqrt{d/t}$; the SoS lower bound follows from
$\distnn_t\le\dist_t$.

The scope of the conclusion is now transparent. The candidate's strict positivity away from the origin ensures
that the sampled values $V(k)$ are positive and that their ratio is meaningful. The argument never asks how
global nonnegativity is certified, which is why the lower bound survives the enlargement from SoS forms to all
nonnegative forms.
\end{proof}

\section{The construction and its facet generalization}
\label{sec:upper-alg}

\subsection{Averaging the facet orbit}

The oriented facet inequalities of the cross-polytope are $\varepsilon^\top x\le1$ for
$\varepsilon\in\{\pm1\}^d$. Equivalently, modulo $\varepsilon\sim-\varepsilon$, opposite pairs form the slabs
$|\varepsilon^\top x|\le1$. Because the power is even, averaging over all oriented normals merely duplicates
each opposite pair and is the canonical construction preserving every symmetry of the body. The Rademacher
variable below records the value of this orbit at the diagonal.

\LemUpper*

\begin{proof}
Let
\[
  f(x)=\E_{\varepsilon}[(\varepsilon^\top x)^{2t}],
  \qquad
  S=\set{x:\gauge{x}{1}=1}.
\]
The form is an average of squares of degree-$t$ forms and is therefore SoS. It is also invariant under signed
permutations.

\emph{Extrema on the $\ell_1$-sphere.}
It is enough to work on the simplex
$\Delta=\set{x\ge0:\sum_i x_i=1}$. Fix two coordinates, write them as
$(m+\delta,m-\delta)$, and freeze the others. If
$c_\varepsilon=\sum_{\ell\ne i,j}\varepsilon_\ell x_\ell$, then
\[
  g(\delta)
  =\E_\varepsilon\!\left[
    \big(c_\varepsilon+m(\varepsilon_i+\varepsilon_j)
      +\delta(\varepsilon_i-\varepsilon_j)\big)^{2t}
  \right]
\]
is convex in $\delta$. Swapping $\varepsilon_i$ and $\varepsilon_j$ shows that it is even, hence nondecreasing
in $\abs\delta$. Equalizing two coordinates can only decrease $f$, so $f$ is Schur-convex on $\Delta$. The
uniform point $d^{-1}\mathbf 1$ is majorized by every point in $\Delta$, while $e_1$ majorizes every such point.
Consequently,
\[
  \max_{x\in S}f(x)=f(e_1)=1,
  \qquad
  \min_{x\in S}f(x)
  =f(d^{-1}\mathbf 1)
  =d^{-2t}\E[\Rad_d^{2t}].
\]
This proves the exact ratio
\[
  \ratio(f^{1/(2t)})
  =\frac{d}{\big(\E[\Rad_d^{2t}]\big)^{1/(2t)}}.
\]

The two-coordinate calculation has a simple geometric reading. Moving $\ell_1$ mass from a larger coordinate
to a smaller one lowers the even signed moment. Repeating these elementary balancing transfers---often called
Robin Hood transfers---leads to the diagonal; reversing them leads to a vertex. This is why no other boundary
directions enter the exact ratio.

\emph{Moment estimates.}
The geometry has now been reduced to determining the scale of one moment. To obtain distortion
$\sqrt{d/t}$, this moment should have scale $(dt)^t$. Gaussian replacement gives its ceiling, while the
collision-free pairing terms already give its floor.
The Rademacher moment satisfies
\begin{equation}\label{eq:moment-sandwich}
  (2t-1)!!\,d^{\downarrow t}
  \ \le\ \E[\Rad_d^{2t}]
  \ \le\ (2t-1)!!\,d^t ,
\end{equation}
where $d^{\downarrow t}=d(d-1)\cdots(d-t+1)$. For the upper bound, replace the signs one at a time by standard
Gaussians. In the binomial expansion, the two variables have the same moments in degrees zero and two, while
every higher even Gaussian moment dominates its Rademacher counterpart. Iteration gives
\[
  \E[\Rad_d^{2t}]
  \le\E[N(0,d)^{2t}]
  =(2t-1)!!\,d^t.
\]
For the lower bound, expand the moment and retain only those monomials in which $t$ distinct indices each occur
twice. There are $(2t-1)!!\,d^{\downarrow t}$ such monomials, each with expectation one; all other surviving
monomials are nonnegative. This proves \eqref{eq:moment-sandwich}.

The upper half of \eqref{eq:moment-sandwich} gives
\[
  \ratio(f^{1/(2t)})
  \ge\frac{\sqrt d}{((2t-1)!!)^{1/(2t)}}.
\]
For a uniform upper bound, use
$(2t-1)!!\ge t!\ge(t/e)^t$ and
$d^{\downarrow t}=t!\binom dt\ge(d/e)^t$. The lower half of
\eqref{eq:moment-sandwich} then gives
\[
  \ratio(f^{1/(2t)})\le e\sqrt{d/t}.
\]
If $\binom t2<d$, the elementary inequality
$d^{\downarrow t}\ge d^t(1-\binom t2/d)$ gives the sharper estimate in
\cref{lem:upper}. Finally, when $t=o(\sqrt d)$,
$d^{\downarrow t}=d^t(1+o(1))$; together with
$((2t-1)!!)^{1/(2t)}=\sqrt{2t/e}\,(1+o(1))$, this yields the stated asymptotic value.

Thus the factor $\sqrt t$ gained over the quadratic approximation is already visible combinatorially: it is
the $2t$-th root of the number of pairings that survive in the moment expansion.
\end{proof}

\subsection{Optimizing a restricted construction}

The sign design suggests allowing nonuniform weights on the facet normals of a general polytope. Every term
$(a_i^\top x)^{2t}$ is at most one on the boundary, so the normalized construction has a built-in ceiling. The
quantity that remains to protect is its floor: maximize the smallest boundary value over the weights. This is
the origin of $\Phi_t$ and of the one-sided certificate $\designbd_t$.

The associated matrix
\[
  \Mt(w)=\sum_iw_i\ver_t(a_i)\ver_t(a_i)^\top
\]
is the information matrix for degree-$t$ polynomial features, which makes classical optimal-design tools
available. The objective here is nevertheless different from the full SoS problem: it optimizes the boundary
floor of one restricted family of facet-power forms.

\ThmAlgorithm*

\begin{proof}
Write $c_i(x)=(a_i^\top x)^{2t}$, so that
$f_{t,w}(x)=\inner{c(x)}{w}$.

\emph{\ref{it:concave}: concavity and soundness.}
The function
\[
  \Phi_t(w)=\min_{x\in\partial\Poly}\inner{c(x)}{w}
\]
is a pointwise minimum of linear functions of $w$, hence concave. The boundary is compact, so the minimum is
attained. Moreover $0\le c_i(x)\le1$ on $\partial\Poly$, which gives $0\le\Phi_t(w)\le1$.
In fact,
$|\Phi_t(w)-\Phi_t(v)|\le\norm{w-v}_1$, so $\Phi_t$ is continuous on the compact simplex and its maximum is
attained. At the uniform weight vector, spanning of the facet normals makes $f_{t,w}$ positive away from the
origin; compactness of $\partial\Poly$ then gives $\Phi_t(w)>0$. Thus $\designbd_t(\Poly)$ is finite.

Every $f_{t,w}$ is SoS. If
$m(w)=\min_{\partial\Poly}f_{t,w}$ and
$M(w)=\max_{\partial\Poly}f_{t,w}$, then $M(w)\le1$ and
\[
  \ratio_\Poly(f_{t,w}^{1/(2t)})
  =\left(\frac{M(w)}{m(w)}\right)^{1/(2t)}
  \le\Phi_t(w)^{-1/(2t)}.
\]
Taking infima proves
\[
  \distnn_t(\Poly)\le\dist_t(\Poly)
  \le\inf_w\ratio_\Poly(f_{t,w}^{1/(2t)})
  \le\designbd_t(\Poly).
\]
The last step has two potential sources of slack: it replaces the actual maximum $M(w)$ by the bound one, and
it restricts the construction to facet-power forms rather than all SoS forms.

\emph{\ref{it:mwu}: mirror ascent with an approximate oracle.}
Suppose $x_s\in\partial\Poly$ satisfies
$f_{t,w_s}(x_s)\le\Phi_t(w_s)+\delta$, and put
$g_{s,i}=c_i(x_s)$. For every $w\in\simplex_n$,
\[
  \Phi_t(w)
  \le\inner{g_s}{w}
  \le\Phi_t(w_s)+\inner{g_s}{w-w_s}+\delta.
\]
Thus $g_s$ is a $\delta$-supergradient, and
$\norm{g_s}_\infty\le1$. For the exponential update, any $w\in\simplex_n$ satisfies
\[
  D_{\mathrm{KL}}(w\Vert w_{s+1})-D_{\mathrm{KL}}(w\Vert w_s)
  =-\eta\inner{g_s}{w}+\log\sum_iw_{s,i}e^{\eta g_{s,i}}
  \le-\eta\inner{g_s}{w-w_s}+\frac{\eta^2}{8},
\]
For the last inequality, the second derivative of the log moment is the variance of a variable in $[0,1]$,
hence at most $1/4$; integrating twice gives the remainder $\eta^2/8$. Telescoping from the uniform distribution
and taking $\eta=\sqrt{2\log n/T}$ gives, for $n\ge2$,
\[
  \sum_{s=1}^T\inner{g_s}{w-w_s}
  \le\frac{\log n}{\eta}+\frac{\eta T}{8}
  \le\sqrt{2T\log n}
  \qquad\text{for every }w\in\simplex_n
\]
for the normalized update in \cref{thm:algorithm}; this is the negative-entropy specialization of the standard
mirror-descent inequality~\citep[Sec.~4.6, Eq.~(4.10)]{Bubeck2015ConvexOpt}. Averaging the
approximate-supergradient inequalities and using concavity,
\[
  \max_w\Phi_t(w)-\Phi_t(\bar w_T)
  \le\sqrt{\frac{2\log n}{T}}+\delta.
\]
If the oracle errors vary with $s$, the same argument replaces $\delta$ by
$T^{-1}\sum_s\delta_s$; the case $n=1$ is immediate.

Geometrically, the minimizing boundary point exposes a direction that the current weights protect least well.
Its coordinates $c_i(x_s)$ measure how strongly each facet functional sees that direction, and the exponential
update reweights according to this exposure. The regret inequality certifies the average effect; it does not
assert that every individual update increases $\Phi_t$.

\emph{\ref{it:oracle}: a convex representation of the oracle.}
The equality constraint $\gauge{x}{\Poly}=1$ makes the boundary search appear nonconvex. Homogeneity removes
that fixed scale: the minimum of one norm on the boundary of another is the reciprocal of a maximum over a norm
ball. Since the polytope gauge is itself the maximum of the facet functionals, the problem then separates into
finitely many dual-norm evaluations.
For positive weights, $h_w=f_{t,w}^{1/(2t)}$ is a norm because the facet normals span $\R^d$. By homogeneity,
\[
  \min_{\gauge{x}{\Poly}=1}h_w(x)
  =
  \left(\max_{h_w(x)\le1}\gauge{x}{\Poly}\right)^{-1}
  =
  \left(\max_i h_w^*(a_i)\right)^{-1}.
\]
Raising to the power $2t$ proves the identity in \cref{thm:algorithm}. Each dual norm value is the convex
program
\[
  h_w^*(a_i)=
  \max_x\set{a_i^\top x:\sum_jw_j(a_j^\top x)^{2t}\le1}.
\]
If $i$ and $x$ attain the two maxima, then
$h_w(x)=1$ and dual H\"older gives
$\abs{a_j^\top x}\le h_w^*(a_j)\le h_w^*(a_i)=a_i^\top x$ for every $j$. Hence
$\gauge{x}{\Poly}=a_i^\top x$, and
$x/\gauge{x}{\Poly}$ is a minimizing boundary point. The reduction therefore gives both the oracle value and an
oracle direction.

\emph{\ref{it:ccly}: the quadratic case.}
At $t=1$, every object in the construction should reduce to classical ellipsoidal design. Indeed, the dual
norm values below become leverage scores, and Kiefer--Wolfowitz closes the loop.
When $t=1$, $h_w(x)^2=x^\top M(w)x$ and
$h_w^*(a_i)^2=a_i^\top M(w)^{-1}a_i$. Hence
\[
  \Phi_1(w)=
  \frac{1}{\max_i a_i^\top M(w)^{-1}a_i}.
\]
For every nonsingular design, the sensitivities
$\ell_i(w)=a_i^\top M(w)^{-1}a_i$ satisfy
\[
  \sum_iw_i\ell_i(w)=\tr\!\big(M(w)^{-1}M(w)\big)=d,
\]
so $\max_i\ell_i(w)\ge d$. Conversely, let $w^\star$ maximize $\log\det M(w)$. The directional derivative
toward the $i$th simplex vertex is $\ell_i(w^\star)-d\le0$, hence every sensitivity is at most $d$.
Concavity of $\log\det$ proves the converse as well: any design with $\max_i\ell_i(w)\le d$ is $D$-optimal.
This is the finite-design proof of the Kiefer--Wolfowitz equivalence theorem~\citep[p.~364]{KieferWolfowitz1960}.
If $M(w)$ is singular, a nonzero kernel vector can be scaled to $\partial\Poly$ and gives $\Phi_1(w)=0$;
therefore every maximizer is nonsingular. We have proved
\[
  \min_{w\in\simplex_n}\max_i a_i^\top M(w)^{-1}a_i=d
\]
and identified its minimizers with the $D$-optimal designs. If $w^\star$ is one of them, then
\[
  E_{w^\star}=\set{x:d\,x^\top M(w^\star)x\le1}
\]
is the symmetric John ellipsoid of $\Poly$. Indeed, writing an ellipsoid as
$E_X=\set{x:x^\top X^{-1}x\le1}$, the primal--dual KKT relation is
$X^{-1}=dM(w^\star)$. The sensitivity inequalities give feasibility for
$X=(dM(w^\star))^{-1}$, and the identity above forces equality whenever $w_i^\star>0$; the multipliers
$d w_i^\star$ therefore give complementary slackness.

In simplex normalization, the iteration of \citet{CohenCousinsLeeYang2019} is
\[
  w_i^+
  =\frac{w_i\,\ell_i(w)}{d},
  \qquad
  \ell_i(w)=a_i^\top M(w)^{-1}a_i.
\]
Their weights sum to $d$; the displayed formula follows after the change of variables $v=dw$.
It optimizes the same $D$-optimal design problem, but it uses the log-determinant sensitivity $\ell_i(w)$ rather
than the worst-direction supergradient $c_i(x_s)$ above. The optimizer is shared; the update is not.
\end{proof}

\subsection{What the facet model misses}

The generalization above creates a tempting but false inference: perhaps optimizing the symmetric facet-power
family already optimizes over all SoS forms. The smallest nonquadratic case, $d=t=2$, is enough to disprove it.

\PropDesignGap*

\begin{proof}
For the square, consider
\[
  p(x,y)=x^4+y^4-x^2y^2=(x^2-y^2)^2+(xy)^2.
\]
This is an SoS form. On the boundary of $B_\infty^2$, one coordinate has squared magnitude one; if the other has
squared magnitude $z\in[0,1]$, then
\[
  p=1-z+z^2\in[3/4,1].
\]
Therefore $\dist_2(B_\infty^2)\le(4/3)^{1/4}$.

The two opposite facet pairs of the square are represented by $e_1,e_2$. For
$w=(w_1,w_2)$,
\[
  \Phi_2(w)=\min_{\max\{|x|,|y|\}=1}(w_1x^4+w_2y^4)
  =\min\{w_1,w_2\}.
\]
Its maximum over the simplex is $1/2$, so
$\designbd_2(B_\infty^2)=2^{1/4}$.
Moreover, the boundary point $(1,1)$ gives $w_1+w_2=1$ for every $w$, so the boundary maximum is exactly one.
Thus $2^{1/4}$ is also the best actual ratio inside this facet-power family; the gap below is not caused by the
ceiling relaxation.

Finally, the linear map $(x,y)\mapsto(x+y,x-y)$ sends the $\ell_1$-norm to the $\ell_\infty$-norm. Distortion
and the corresponding facet certificate are invariant under invertible linear changes of variables. Indeed, if
$A$ is invertible and $b_i=A^{-\top}a_i$ are the transformed facet normals, then
\[
  \sum_iw_i(b_i^\top Ax)^{2t}=\sum_iw_i(a_i^\top x)^{2t}.
\]
The map $x\mapsto Ax$ therefore preserves the boundary extrema and $\designbd_t$, while polynomial pullback
gives the same invariance for $\dist_t$. This yields the same strict gap for $B_1^2$. Explicitly, the transformed
SoS form is
$x^4+14x^2y^2+y^4$.

The conceptual point is the mixed term. A nonnegative mixture of coordinate facet powers cannot use a negative
$x^2y^2$ coefficient, whereas the full SoS cone can: the interaction is absorbed by the square decomposition.
Symmetry controls which terms may appear, but it does not collapse these two cones.
\end{proof}

\begin{remark}[What symmetry still gives]
For the cube, symmetry gives
$\designbd_t(\Binf)=d^{1/(2t)}$, but \cref{prop:design-gap} shows that this need not equal
$\dist_t(\Binf)$. For the cross-polytope, concavity and group averaging make the uniform sign design optimal for
$\designbd_t$; hence, for $1\le t\le d$,
\[
  \designbd_t(\Bone)
  =\frac{d}{(\E[\Rad_d^{2t}])^{1/(2t)}}
  =\Theta(\sqrt{d/t}).
\]
In this same range, the main theorem makes this certificate order-optimal, not exactly optimal, over the full
SoS cone.
\end{remark}

\section{Discussion}
\label{sec:discussion}

\paragraph{What the proof teaches.}
The cross-polytope obstruction is carried by orbit geometry, not by a particular representation of
nonnegativity. After averaging, dimension appears as the number of support scales
$k=1,\ldots,d$ on which a candidate must remain well balanced; degree appears as the complexity of the profile
available to balance them. The flat points are effective not because they are known to contain every extremum,
but because they make invariant forms forget almost everything: each restriction is
$k^{-2t}Q(k)$ with $\deg Q\le t$. Interpolation then converts the mismatch between $d$ scales and a degree-$t$
numerator into the factor $\sqrt{d/t}$.

This view also explains the strongest scope feature of the theorem. Once the sampled values are positive, the
proof no longer sees whether the form is a sum of squares, a sum of powers, or merely globally nonnegative. The
lower bound is therefore representation-free. The construction side uses the same symmetry differently: the
facets form one sign-vector orbit, whose uniform power average produces the Rademacher form. The lower and
upper bounds are thus two uses of one group---simplify every adversary, then generate one canonical match.

\paragraph{A portable diagnostic.}
The weighted form of \cref{lem:weighted-profile} separates the interpolation argument from the geometry that
produces its exponent. To test another symmetric body, one may look for three ingredients: a group average that
does not increase the ratio; boundary representatives indexed by a support parameter $k\in[d]$; and a collapse
of every averaged degree-$2t$ form to
\[
  V(k)=k^{-\alpha}Q(k),
  \qquad \deg Q\le t,
  \qquad Q(0)=0.
\]
If these ingredients hold and $\alpha\ge t$, positivity on the representatives alone forces
\[
  \frac{\max_kV(k)}{\min_kV(k)}
  \ge 2(2e)^{-t}\left(\frac dt\right)^{\alpha-t}.
\]
This is a test rather than a universal theorem about symmetric bodies: the body-specific work is precisely to
find an orbit family on which the invariant algebra collapses in this way. For $B_1^d$, $\alpha=2t$; for the
$\ell_p$ witnesses, $\alpha=2t/p$. Thus the same lemma both proves the cross-polytope rate and locates the
$p=2$ stopping point.

\paragraph{A correct-order certificate need not be an optimizer.}
The cube and the facet program give the same warning in two forms. The coordinate certificate
$\sum_i x_i^{2t}$ proves the cube--cross-polytope separation, but the quartic example in
\cref{prop:design-gap} shows that it is not always optimal. Likewise, depending on the body, $\designbd_t$ may
be looser than the best ratio in its own facet-power family; independently, that restricted family may be
looser than the optimum over the full SoS cone.

For a weight vector $w$, the actual facet-form ratio is
\[
  \left(\frac{M(w)}{m(w)}\right)^{1/(2t)},
  \qquad
  m(w)=\min_{\partial\Poly}f_{t,w},
  \quad
  M(w)=\max_{\partial\Poly}f_{t,w},
\]
whereas $\designbd_t$ optimizes $m(w)^{-1/(2t)}$ after using $M(w)\le1$. The normalization is one possible
source of slack, and restricting to facet powers creates a second; \cref{prop:design-gap} proves that the latter
occurs already for a symmetric quartic example. Indeed, for the square $M(w)=1$ for every $w$, since the
boundary point $(1,1)$ gives $w_1+w_2=1$; the strict gap in \cref{prop:design-gap} therefore comes from the
restricted form family, not from the ceiling relaxation.

The normalization slack also occurs by itself. For the regular hexagon described by three unit normals at
angles $0,\pi/3,2\pi/3$, uniform weights give information matrix $I/2$ and boundary floor $1/2$. By
\cref{thm:algorithm}, this design is optimal for the boundary-floor program, so $\designbd_1=\sqrt2$. Its
boundary maximum is $2/3$, however, and the same facet form has actual ratio $2/\sqrt3<\sqrt2$. Dihedral
averaging cannot increase a quadratic sandwich ratio and sends every quadratic form to a scalar multiple of
$\norm{x}_2^2$, so this Euclidean form is also optimal over the full quadratic cone. The example therefore
isolates the ceiling relaxation. The value of the program is that its guarantee, optimizer, oracle, and failure
mode are all explicit; it should not be identified with $\dist_t$.

The quadratic endpoint is a useful consistency check, with one similar qualification. Maximizing $\Phi_1$ is
the $G$-optimal problem and produces a $D$-optimal design; after the factor-$d$ scaling in
\cref{thm:algorithm}, its information matrix defines the John ellipsoid. But
$\designbd_1(\Poly)=\sqrt d$ universally, while a particular body may have a better surrounding ellipsoidal
ratio. What survives exactly at $t=1$ is the optimizer and the inscribed ellipsoid, not equality among every
approximation objective.

\paragraph{Sharp constants: which side must move?}
For $t\to\infty$ with $t=o(\sqrt d)$, the sign design has asymptotic constant $\sqrt{e/2}$, whereas the
interpolation lower bound has constant $1/\sqrt{2e}$. The factor-$e$ gap has two plausible sources, and the
present argument does not choose between them. On the lower side, we bound only the $\ell_1$ norm of the
Lagrange evaluation functional and then relax a factorial. On the upper side, the quartic design gap warns that
the uniform sign form may not be optimal in the full SoS cone. A sharper grid extremal inequality would move
the first side; a construction using mixed SoS interactions would move the second. Determining which happens
is the sharp-constant problem exposed by the proof.

\paragraph{The cube: determine the optimum, not another certificate.}
For the cube we know
$\dist_t(B_\infty^d)\le d^{1/(2t)}$, which is enough for logarithmic-degree constant approximation. The next
question is the order of the full SoS optimum for fixed $t>1$. The two-dimensional quartic improvement rules
out a proof that simply declares the coordinate certificate optimal by symmetry. Any successful lower bound
must control the mixed invariant terms that facet powers omit.

\paragraph{Beyond $p=2$: replace the exponent-sensitive step.}
For $1\le p\le2$, the flat-point exponent lies on the side where the same interpolation inequality yields a
lower bound. At $p>2$ it crosses the degree of the numerator, and the direction of that comparison no longer
controls oscillation. This is a structural stopping point of the proof, not a prediction of the answer. A
concrete first target is the interval $2<p<4$ at quartic degree, where the resonance at $p=4$ and the exact
quadratic law provide two anchors. Progress there requires a replacement for the exponent-sensitive
interpolation step, not more bookkeeping within the present estimate.

\paragraph{Computation: protect a relative objective.}
The outer mirror-ascent guarantee is additive in $\Phi_t$. The inner problem has an exact convex dual-norm
formulation, but solving its $n$ convex programs may dominate the running time, and the map
$\Phi_t\mapsto\Phi_t^{-1/(2t)}$ is ill-conditioned when the optimum is small. A useful algorithmic target is a
relative-error guarantee under a quantitative lower bound on $\max_w\Phi_t(w)$, together with an oracle whose
complexity tracks the conditioning of the induced norm. These questions concern the certificate program; they
do not affect the geometric lower bound.

The research map is therefore modular: better constants require a sharper profile inequality or a new SoS
construction; the exact cube problem requires control beyond facet powers; the region $p>2$ requires a new
witness mechanism; and faster computation requires conditioning-aware dual norms. The proof does not merely
leave these problems open; it locates each one at the step where the present framework loses the information a
sharper result would need.

\paragraph{AI Disclosure.}
The author acknowledges the use of GPT-5.6 Solar as an assistive tool in preparing this manuscript,
including support in drafting, revising, and polishing the exposition. All AI-assisted material was
carefully reviewed and edited by the author, who takes full responsibility for the final manuscript
and its results.

\bibliographystyle{plainnat}   
\bibliography{refs}

\appendix
\section{Deferred proofs}
\label{app:proofs}

\subsection{Assembly of the main theorem}

\ThmMain*

\begin{proof}
The lower-bound argument at the end of \cref{sec:lower} applies to every globally nonnegative form, and therefore
gives
\[
  \frac1{\sqrt{2e}}\sqrt{\frac dt}
  \le\distnn_t(\Bone)
  \le\dist_t(\Bone).
\]
Barvinok's theorem supplies the first SoS upper bound in \eqref{eq:main-bound}, and \cref{lem:upper} supplies the
second. Finally,
\[
  \binom{d+t-1}{t}
  \le\left(\frac{e(d+t-1)}t\right)^t
  \le\left(\frac{2ed}t\right)^t
\]
for $1\le t\le d$, which proves the last inequality in \eqref{eq:main-bound}. The more precise asymptotic value
of the sign design was established in \cref{sec:upper-alg}.

For the first consequence, the universal SoS upper bound is
$O(\sqrt{d/t})$ for every symmetric body, while the cross-polytope has the matching lower bound even over the
larger nonnegative cone.

For the second consequence, if $t\le d$ and either distortion is at most $C$, then
$t\ge d/(2eC^2)$; if $t>d$, this lower bound on $t$ is automatic. Conversely, the all-degree Barvinok bound
implies that, for every fixed $C>1$, some
$t=O_C(d)$ gives SoS distortion at most $C$, and hence also nonnegative-form distortion at most $C$. To see the
linear dependence directly, the case $d=1$ being immediate, take $t=\lceil cd\rceil$ for $d\ge2$ and use
\[
  \binom{d+t-1}{t}
  =\binom{d+t-1}{d-1}
  \le\left(\frac{e(d+t-1)}{d-1}\right)^{d-1};
\]
the $1/(2t)$ power tends to one as the constant $c$ tends to infinity.
\end{proof}

\subsection{The polar-body comparison}

\CorSeparation*

\begin{proof}
The cube certificate in \cref{fact:cube-upper} gives
\[
  t_C^*(\Binf)
  \le\left\lceil\frac{\log d}{2\log C}\right\rceil.
\]
For the cross-polytope, the lower bound in \cref{thm:main} gives
$t_C^*(\Bone)\ge d/(2eC^2)$, while the all-degree Barvinok upper bound used above gives
$t_C^*(\Bone)=O_C(d)$. Thus $t_C^*(\Bone)=\Theta_C(d)$.

If $t=\lceil\log_2d\rceil$, then $d^{1/(2t)}\le\sqrt2$, so the cube bound follows from
\cref{fact:cube-upper}; the cross-polytope bound is \eqref{eq:main-bound}.
\end{proof}

\subsection{\texorpdfstring{The $\ell_p$ family}{The ell-p family}}

\ThmEllp*

\begin{proof}
\emph{\ref{it:reson}: resonances.}
Let $p=2j$ and suppose $j\mid t$. Then
\[
  q(x)=\left(\sum_i x_i^{2j}\right)^{t/j}
      =\gauge{x}{p}^{2t}.
\]
The form is SoS: expanding the integer power of
$\sum_i(x_i^j)^2$ expresses it as a sum of squares of degree-$t$ monomials. Its ratio is one, which is the
smallest possible value for either polynomial cone.

\emph{\ref{it:t1law}: degree two.}
At $t=1$, a nonnegative quadratic form is SoS. Averaging over signed permutations does not increase the ratio and
turns its matrix into a scalar multiple of the identity. The ratio of the Euclidean norm to the $\ell_p$ norm on
the $\ell_p$-sphere is
$d^{\abs{1/2-1/p}}$, proving the stated equality.

\emph{\ref{it:plaw}: the range $1\le p\le2$.}
Average an arbitrary nonnegative form over signed permutations and evaluate it at
\[
  x^{(k)}=k^{-1/p}(\underbrace{1,\ldots,1}_{k},0,\ldots,0).
\]
Since
$p_{2r}(x^{(k)})=k^{1-2r/p}$, the resulting flat-point values have the form
\[
  V(k)=k^{-\alpha}Q(k),
  \qquad
  \alpha=\frac{2t}{p},
  \qquad
  \deg Q\le t,\quad Q(0)=0.
\]
This is the same support-size profile as for the cross-polytope, with one controlled change: the geometric
normalization replaces the exponent $2t$ by $\alpha=2t/p$. The inequality $p\le2$ is exactly the condition
$\alpha\ge t$ required by \cref{lem:weighted-profile}. Applying that lemma gives
\[
  \frac{\max_kV(k)}{\min_kV(k)}
  \ge
  2(2e)^{-t}\left(\frac dt\right)^{\alpha-t}.
\]
After taking the $2t$-th root and dropping the harmless factor $2^{1/(2t)}$,
\[
  \distnn_t(\Bp)
  \ge\frac1{\sqrt{2e}}
       \left(\frac dt\right)^{1/p-1/2}.
\]
The SoS form $(\sum_i x_i^2)^t$ has $2t$-th root $\norm{x}_2$, whose ratio on the $\ell_p$-sphere is
$d^{1/p-1/2}$. This proves the two-sided chain. If $p<2$ is fixed and
$\log t=o(\log d)$, the logarithms of the lower and upper bounds are both
$(1/p-1/2+o(1))\log d$, proving the final assertion.
\end{proof}

\begin{remark}
The three parts of \cref{thm:ellp} arise from three different mechanisms: exact algebraic representation at the
resonance points, uniqueness of the invariant quadratic at $t=1$, and profile interpolation for $p\le2$. They
should not be read as samples from a single smooth formula.

The upper bound for $1<p<2$ does not recover the $t$-dependence in the lower bound. A matching construction would
require an analogue of the sign design adapted to $\ell_p$. For $p>2$, even the correct power of $d$ away from
the resonance points remains open.
\end{remark}

\end{document}